\documentclass[dvipsnames,
                format=sigconf,
                anonymous=false,
                authorversion=false,
                review=false]{acmart}

\usepackage{listings}

\AtBeginDocument{%
  \providecommand\BibTeX{{%
    \normalfont B\kern-0.5em{\scshape i\kern-0.25em b}\kern-0.8em\TeX}}}

\setcopyright{acmcopyright}
\copyrightyear{2023}
\acmYear{2023}

\acmConference[GECCO '23]{The Genetic and Evolutionary Computation Conference 2023}{July 15--19, 2023}{Lisbon, Portugal}

\begin{document}

\title{Byzantine-Resilient Learning Beyond Gradients: \\ Distributing Evolutionary Search}

\author{Andrei Kucharavy}
\email{andrei.kucharavy@hevs.ch}
\affiliation{
  \institution{IC School, EPFL}
  \streetaddress{Rte Cantonale}
  \city{Lausanne}
  \country{Switzerland}
  \postcode{1015}
}
\authornote{Corresponding author; Now at HES-SO Valais-Wallis, Sierre, Switzerland}

\author{Matteo Monti}
\email{matteo.monti@epfl.ch}
\affiliation{
  \institution{IC School, EPFL}
  \streetaddress{Rte Cantonale}
  \city{Lausanne}
  \country{Switzerland}
  \postcode{1015}
}

\author{Rachid Guerraoui}
\email{rachid.guerraoui@epfl.ch}
\affiliation{
  \institution{IC School, EPFL}
  \streetaddress{Rte Cantonale}
  \city{Lausanne}
  \country{Switzerland}
  \postcode{1015}
}

\author{Ljiljana Dolamic}
\email{ljiljana.dolamic@ar.admin.ch}
\affiliation{
  \institution{Cyber-Defence Campus, armasuisse}
  \streetaddress{Feuerwerkerstrasse 39}
  \city{Thun}
  \country{Switzerland}
  \postcode{3602}}

\renewcommand{\shortauthors}{Kucharavy, et al.}

\begin{abstract}
Modern machine learning (ML) models are capable of impressive performances across numerous domains. However, their prowess is not due only to the improvements in their architecture and training algorithms, but also to the drastic increase in computational power available to them.

Such a drastic increase led to a growing interest in distributed ML algorithms, and, as adversarial attacks became an increasingly pressing concern, the emergence of their byzantine-resilient versions. While distributed byzantine resilient algorithms have been proposed in a differentiable setting, none exist in a gradient-free setting.

The goal of this paper is to address this shortcoming. For that, we first introduce a more general definition of byzantine-resilience in ML - the \textit{model-consensus}, that extends the definition of the classical distributed consensus. We then leverage this definition to show that a general class of gradient-free ML algorithms - ($1,\lambda$)-Evolutionary Search - can be combined with classical distributed consensus algorithms to generate gradient-free byzantine-resilient distributed learning algorithms. We provide proofs and pseudo-code for two specific cases - the Total Order Broadcast and proof-of-work leader election.

To our knowledge, this is the first time a byzantine resilience in gradient-free ML was defined, and algorithms to achieve it - were proposed.
\end{abstract}




\maketitle

\section{Introduction}

Over the last decade, the machine learning field underwent transformative growth, achieving and surpassing human capabilities in a variety of domains, ranging from image classification and facial recognition to image generation to strategy games \cite{ImageNet2012, 2015FaceNet, Karras2019, AlphaGoTwo2020}. Beyond impressive performance in the academic setting, Machine Learning (ML) and Artificial Intelligence (AI) progressively became central to numerous tasks, ranging from translation to autonomous driving \cite{GoogleTranslate2017JeffDean, DeepLearningSelfDriving2020}. Perhaps the most impressive recent development is the arrival of conversational agents driven by Large Neural Language Models (LLMs) \cite{InstructGPT2022OpenAI, BlenderBot2022Facebook, Sparrow2022Google}.

However, the emergence of ML and AI as powerful and widely accessible tools is not only due to the discovery of better model architectures and algorithms to train them but also due to the increasing computational capabilities and data volumes available to train them. Empirical demonstrations of the performance of stochastic gradient descent (SGD) applied to artificial neural networks (ANNs) were already available by mid-1980s \cite{Werbos1974, LeCun1985}, and theoretically explained by early 1990s \cite{Hornik1989}. However, it wasn't until sufficient computational power became affordable and sufficiently large training datasets were accumulated that the machine learning revolution truly started \cite{lecun2015}. This joint scaling of models and dataset sizes and resources invested in training them still drives ML progress today, notably for LLMs \cite{ScalingLawsLLM2020OpenAI, Antropic2022BiggerModels, ComputeOptimalLLMs2022Google}.

\subsection{Gradient-Free Learning}

The ongoing machine learning revolution has not affected all the domains equally, given that best-performing algorithms rely on ANNs and gradient descent. Image processing was one of the first domains that saw early breakthroughs \cite{ImageNet2012}, more recently followed by natural language processing as means to make text interpretation and generation continuous through word embedding and positional encoding were perfected \cite{Word2Vec2013, Transformers2017}. However, a number of problems have so far evaded conversion to a continuous formulation, notably in the control theory domain. 

A set of approaches have been developed for such problems, with empirical gradients and reinforcement learning (RL) being two major examples. However, neither of the approaches scales to large, overparameterized models, known to be needed to train models solving complex problems \cite{LossLandscapesVis2018}. Empirical gradient estimation struggles in high dimensions and around saddle points, and is hard to parallelize due to the need for a synchronized round of gradient evaluations pooling, which is expensive computation and communication-wise for larger models. 

Similar problems exist as well in reinforcement learning. Despite some of its impressive achievements, such as super-human performance in the game of Go \cite{AlphaGoTwo2020}, RL suffers from fundamental limitations. Notable failure modes are the cases where the observations ("rewards") are sparse ("long time horizons") and noisy. In such settings, the policy reward estimator's variance will increase to the point where the learning process becomes unstable. Such instability is not limited to pathological settings - even in cases it performs well, RL requires a hyperparameter space search to find a working training regime even for problems where it performs well \cite{SalmansSutskever2017, JeffCluneNeurIPS2018}. Such instability is not a fluke either. There are theoretical reasons why approaches that reduce learning in a non-differentiable setting to a differentiable one would underperform compared to gradient-free black-box optimization approaches \cite{SalmansSutskever2017, GradientsIsNotAllYouNeed2022}. 

This is particularly relevant now, given that the latest development in the LLM field is conversation agents, which rely on optimization based on discrete feedback to align their behavior on user expectations and non-differentiable layers of hard attention to solve the issues with rule-following that plague them \cite{InstructGPT2022OpenAI, SeeKeR2022Facebook, Sparrow2022Google}.



\subsection{Evolutionary Search}

Limitations of RL and empirical gradients approaches led to an increased interest in gradient-free black-box optimization algorithms, notably Evolutionary Algorithms (EAs). Introduced shortly after the SGD itself, \cite{Fogel1966}, EAs are expected to scale well with more computational power, just like the SGD itself. This was confirmed experimentally, including on ranges of control tasks where they outperformed RL approaches, all while allowing better scaling \cite{SalmansSutskever2017, 2017UberGeneticAlgos, JeffCluneNeurIPS2018, CluneMiikkulainen2019}. These empirical results led to a renewed interest in EAs in ML and discovery of cases where they out-performed RL and other black-boxes approaches, such as model architecture design \cite{GoogleEvoImageArchiSearch2019}.

Despite its simplicity, the first evolutionary algorithm proposed by \cite{Fogel1966} in 1966 is still able to match and out-perform RL approaches on complex problems \cite{2017UberGeneticAlgos, CluneMiikkulainen2019, GoogleEvoImageArchiSearch2019}. Akin to SGD, it is an iterative optimization algorithm, however, instead of calculating the local gradient, it samples the neighborhood of the current model parameters to find a better solution and retains the single best one among all sampled ones. Formally known as ($1,\lambda$)-Evolutionary Search (($1,\lambda$)-ES) \footnote{Given multiple conflicting names for different EA algorithms, we have here adopted the taxonomy from \cite{2015EvolutionStrategiesReview}. Notable cases are the use of the "Genetic Algorithm" name to designate ($1,\lambda$)-ES in prior literature, that we reserve to algorithms including "chromosomes" or "recombination" as per the original article \cite{Goldberg1989GAs}, or use of Evolutionary Search for Natural Evolutionary Search (NES), that is closer to empirical gradient approach than ES proper \cite{2008NES}.}

In addition to their reasonable performance, ($1,\lambda$)-ES class algorithms have an additional advantage - scalability. As a population algorithm, every parameter sample can be evaluated independently, and an optimal parameter update - shared among all workers once a desired population of candidate updates has been sampled. Here we focus on the simplest implementation of the ($1,\lambda$)-ES class, which we will refer to as  ($1,\lambda$)-ES for simplicity. A modification of that algorithm by \cite{SalmansSutskever2017} reduces the message size to about a dozen bytes regardless of the model size, by leveraging the fact that random parameter perturbations can be deterministically derived by a pseudo-random number generator from a random seed, meaning that sharing only the random seed is sufficient. Unlike gradient-based learning, ($1,\lambda$)-ES allows any worker to verify the validity of an update proposed by another vector with a single forward pass, which is the property we leverage to combine classical distributed consensus algorithms with  ($1,\lambda$)-ES to create byzantine-resilient distributed versions of  ($1,\lambda$)-ES. 

Finally, since at no point  ($1,\lambda$)-ES requires a back-propagation, it allows for non-differentiable layers, such as hard attention or deterministic rules, to be included in the model architecture \cite{HardAttention2015Bengio}. This makes it interesting even in the setting allowing for gradient-based learning because, unlike differentiable layers, deterministic rules can provide deterministic guarantees on AI model decisions, which is critical in high-stakes applications. In particular, for LLMs, it has a potential of solving the long-term instruction retaining problem, currently limiting their application \cite{SeeKeR2022Facebook, Sparrow2022Google}.

\subsection{Byzantine-Resilience in Machine Learning}

The increasing size of ML models has also made them impossible to train or even run on single machines, making model parallelization and distribution an increasingly pressing issue \cite{Pathways2022GoogleJeffDean}. With the increase of the computing nodes involved in the training, the chances for an arbitrary complex error to occur increase, making fault tolerance a prime concern. In the field of distributed computing, the tolerance to such faults is known as "Byzantine Tolerance", with the name derived from the seminal paper that introduced that type of faults \cite{ByzantineGenerals1982Lamport}.

The field of machine learning allowing for such "Byzantine" fault tolerance led to the emergence of the field of byzantine-resilient machine learning \cite{Mehdi01, Mehdi01Bis}. Unfortunately, the definition introduced in the process is specific to differentiable manifolds and focuses on the setting where every node trains the same model, only has partial access to the data, and shares non-verifiable gradients calculated on that data. By introducing novel gradient aggregation rules (GARs) for the parameter server, they were able to prove that a byzantine fault impact could be limited to at most a deviation of angle $\alpha$ on the final parameter update for the fraction $f$ of Byzantine workers (($\alpha$, $f$)-Byzantine Resilience). 

The reason authors had to introduce a new definition of byzantine-resilience rather than to re-use existing ones, is that the latter are poorly suited to the distributed learning setting. If approached from the \textit{Do-All problem} perspective \cite{DoAllProblemInByzantineContext2005}, parameter update vectors cannot be verified without repeating the whole computation, meaning that byzantine-resilience would require several workers to perform redundant update calculation. In the setting where training a single model can cost millions in electricity costs alone (cf eg \cite{GPT3Explained2020OpenAI}), direct redundancy is an unrealistic assumption.

In the distributed gradient-based learning, the ($\alpha$, $f$)-Byzantine Resilience hence remained the predominant paradigm and has been further expanded to provide guarantees for models trained in a more general distributed setting than federated learning \cite{guo2021Mozi, Yang2019Bridge, Yang2019Byrdie, Mehdi2020ByzSGD, Mehdi2021Jungle}.

\subsection{Our Contribution}

Our main contribution is to show that Evolutionary Search is can be adapted to work as a byzantine-resilient distributed optimization algorithm in a non-differentiable setting. 

Specifically, we show that by introducing a new definition of distributed consensus in the ML setting, we can leverage the existing literature on byzantine-resilient distributed computing. In turn, by using the established primitives of the total order broadcast and proof-of-work probabilistic consensus primitives \cite{OGReliableBroadcast1984, 2019NakomotoConsensusAnalysis} we propose two algorithms for distributed evolutionary search - in permissioned (closed) and permissionless (open) settings, and prove the bounds on the computational overhead imposed by distributing the Evolutionary Search.

Interestingly, our new definition of distributed consensus - the \textit{Model-Consensus} generalizes the  ($\alpha$, $f$)-Byzantine Resilience introduced by \cite{Mehdi01, Mehdi01Bis} and directly interfaces with the more general definition of computational consensus.

\section{Preliminaries}

\subsection{Learning Setting} \label{prelim:LearningSetting}

Our problem consists in learning a general function $f \in \mathbb{F}$, mapping inputs $x \in \mathbb{X}$ to outputs $y \in \mathbb{Y}$, determined by parameters $\boldsymbol{\theta}$, noted $f(. ,\boldsymbol{\theta})$\footnote{For machine learning, we follow the common notation introduced in \cite{DeepLearninBook2016Goodfellow}.}. A scalar performance metric $\mathcal{L}$ is associated to the function and can be compute for each input/output pair in the training and validation sets $\mathcal{L}(f(x, \boldsymbol{\theta}), y)$. We denote $\mathcal{L}_{\theta}$,
an aggregate performance metric on all input-output pairs. Without loss of generality, $\mathcal{L}$ can correspond conversely to loss, accuracy, total reward, fitness, or another metric of the model performance. The goal of learning is to find parameters that optimize that value. This process can be referred to as parameter optimization, parameter space search, or training. For the sake of simplicity, we adopt the convention that $\mathcal{L}$ is a loss that we seek to minimize, although, in the context of evolution, $-\mathcal{L}$ will be occasionally referred to as fitness due to historical reasons. $\{\mathcal{L}_{\theta}\}$ will be referred to as loss landscape (conversely fitness landscape). Finally,  $v\mathcal{\boldsymbol{L}}_{\theta}$ will refer to the \textit{loss vector} obtained by concatenating the losses for all the input/output pairs $(x,y) \in \mathbb{X} \times \mathbb{Y}$ in the training set for model parameters $\boldsymbol{\theta}$.
While we assume a non-differentiable setting, we still assume a smooth loss landscape, ie $\exists k \in \mathbf{R}$ such that $\forall(\boldsymbol{\theta}_i, \boldsymbol{\theta}_j), \frac{|\mathcal{L}_{\theta_i}-\mathcal{L}_{\theta_j}|}{\lVert \boldsymbol{\theta}_i - \boldsymbol{\theta}_j \rVert}<k$.

Given that we are interested in distributing the training phase,  $\mathcal{L}_{\theta}$ without further additional notation designate the aggregated performance metric on the training dataset. We assume as well that every worker has access to all of the training data and that $\mathcal{L}_{\theta}$ is computed in a deterministic manner by each worker, given $\theta$. This setting is different from the one used in distributing the gradient computation, given that the difficulty for the ($1,\lambda$)-ES algorithm is to find a valid update.


\subsection{Model-Consensus and $\epsilon$-Optimality}

In machine learning context, the consensus problem is for a set processes $p\in\Pi$ to \textit{decide} on a common value of model parameters $\boldsymbol{\theta}$ based on model values correct processes can evaluate individually $\boldsymbol{\theta}_p$. A correct process can decide on a value at most once every training session.

For the sake of readability, given the process-worker equivalence, we will be referring to processes $p$ as \textit{workers} and the ensemble of workers trying to solve the machine learning consensus problem for a given task $\Pi$ - as \textit{worker pool}.

A machine learning consensus protocol must satisfy the following conditions: 

\begin{itemize}
  \item \textit{Liveness}: Each correct worker must eventually decide on a value of $\boldsymbol{\theta}$
  \item \textit{Consistency}: No two correct workers can decide on a different  $\boldsymbol{\theta}$.
  \item \textit{Validity} (Extended): For all correct workers, only a $\boldsymbol{\theta}$ proposed by a correct process can be decided upon.
\end{itemize}

Given that ML model training is distributed to improve parameter space search, we expect the workers to propose different values $\boldsymbol{\theta}_p$, so the extended validity is essential. Moreover, we expect some parameters to correspond to a better loss value, and we want our workers to decide on a value of parameters that leads to the lowest loss possible. This leads us to introduce a new constraint on the model-consensus:
\begin{itemize}
  \item \textit{$\epsilon$-optimality}: If $\boldsymbol{\theta}$ satisfies $\mathcal{L}_{\theta} \leq min_{p \in \Pi} (\mathcal{L}_{\theta_p})+\epsilon$, where $\epsilon \geq 0$, the consensus is \textit{$\epsilon$-optimal}.
\end{itemize}

A special case of $\epsilon$-optimality is the case where $\epsilon=0$, in which case we will refer to the consensus simply as \textit{optimal}. The two algorithms we propose here are optimal for each update with high probability, whereas ($\alpha$, $f$)-Byzantine Resilience \cite{Mehdi01, Mehdi01Bis} is $\epsilon$-optimal with $\epsilon=sin(\alpha) \cdot lr \cdot k_{lipshitz}$, where $lr$ is the effective learning rate and $k_{lipshitz}$ - the Lipschitz constant of the loss landscape.


\subsection{($1,\lambda$)-ES Algorithm}
\label{coreAlgSection}

Similar to SGD, algorithms of the ($1,\lambda$)-ES class search for optimal model parameters $\theta$ through a series of steps performing an empirical descent of the loss landscape \cite{kacprzykEvolution2015}. At each step $i$, the value of the parameters $\theta_i$ is perturbed by a vector $\boldsymbol{\beta}$\footnote{Other works tend to use $\epsilon$ to denote it, whereas we use a greek letter close to the neighborhood notation in topology to avoid confusion with $\epsilon$ of $\epsilon$-optimality} sampled from a normal distribution $\mathcal{N}(0, I)$ and scaled to a learning rate $\sigma$. A number ($k \in [1, .., N]$) of $\boldsymbol{\beta}$ values are tested. The one that improves the model the most ($k_{update} = arg \min_{k}{\mathcal{L}(\theta_i + \sigma \boldsymbol{\beta_k})}$) is retained and becomes the base for the next search ($\theta_{i+1}:=\theta_i + \sigma \boldsymbol{\beta_{k_{update}}}$). We refer to $\sigma \boldsymbol{\beta_{k_{update}}}$ as an \textit{update vector}, $\sigma \boldsymbol{\beta_k}$ as \textit{candidate update vectors} and $\boldsymbol{\theta}_i + \sigma \boldsymbol{\beta_k}$ as \textit{candidate parameters}. No update will occur if no tested vectors have improved loss, so only vectors such as $\mathcal{L}(\theta_i + \sigma \boldsymbol{\beta_k}) < \mathcal{L}(\theta_i) + \nu$, where $\nu \geq 0$ is a parameter controlling for a trade-off between random noise due to sampling and gradient descent - a minimal improvement to be achieved before an update is triggered. We will refer to $\sigma \boldsymbol{\beta_k}$ for which this property is true a \textit{valid update vector}.




\subsection{Adapting ($1,\lambda$)-ES for Distributed Setting}

As we mentioned in the introduction, an important improvement to the ($1,\lambda$)-ES is for nodes to share only the random seeds used to derive candidate update vector deterministically with a pseudo-random numbers generator, allowing update sharing with short messages (given that randoms seeds are $<$16 bytes for most ML libraries), and once bundled with the loss parameter, allows any correct node to verify the proposed candidate update vector. 
In all that follows, we will assume that mode of derivation and refer to such a random seed as a \textit{candidate update vector seed}, noted as $\mathfrak{S}_{\beta_k}$. Formalizing the section above, we assume as well that we have an access to a random generator that is capable to turn a random seed into a non-scaled update vector ($RG: \mathfrak{S}_{\beta_k} \to \beta_k$).

To facilitate the proofs for the permissionless setting, we introduce an additional modification of the ($1,\lambda$)-ES algorithm that is run by the workers $p$. Specifically, to ensure strategy-proofness and more closely match existing proof-of-work, we add a combined hashing of the loss and update vector seed, assumed to be a positive integer below a certain maximal value $\mathfrak{B}_{max}$ (eg. the largest integer that can be encoded with the number of bits in a hash). We refer to the hash of ($\mathfrak{S}_{\beta_k}$, $\mathcal{\boldsymbol{L}}_{\theta_k}$) as \textit{$\theta$-block score} $\mathfrak{B}_{\theta_k}$. In the proof-of-work consensus, it is used as a scheduler for leader election, which is triggered when $\mathfrak{B}_{\theta}<\mathfrak{B}_{target}$, where $\mathfrak{B}_{target}$ is the value set to control the frequency of leader election given the size of the worker pool and the frequency of evaluation.

The pseudo-code for the complete evolutionary search algorithm is presented in the listing \ref{ESCode}.






\begin{lstlisting}[caption={Single Worker Evolutionary Search},
label=list:8-6,
captionpos=b,
float,
abovecaptionskip=\medskipamount,
mathescape=true,
label={ESCode}]
Abstraction: 
  EvolutionarySearcher, Instance es
  
Interface:
  - Request <es.Start | $\theta_i$>: starts search
  - Request <es.Stop>: ends search
  - Indication <es.BestHash | $\theta_i$, $\mathfrak{S}_{\beta_k}$, $\mathcal{L}_{\theta_i+\sigma\beta_k}$>:
    a new seed with a valid hash was found
  - Indication <es.BestLoss| $\theta_i$, $\mathfrak{S}_{\beta_k}$, $\mathcal{L}_{\theta_i+\sigma\beta_k}$>:
    a new valid update vector seed was found
  - Procedure es.follow($\theta_i$, $\mathfrak{S}_{\beta_k}$) -> $\theta_i + \sigma \beta_k$
    derive candidate parameters for a seed
  - Procedure es.evaluate($\theta_i$, $\mathfrak{S}_{\beta_k}$) -> 
  ($\mathfrak{B}_{\theta_i + \sigma \beta_k}$, $\mathcal{L}_{\theta_i + \sigma \beta_k}$): evaluates the
    hash and loss of a candidate seed
    
Algorithm:
  Implements:
    EvolutionarySearcher, instance es;
  Parameters:
    $\mathcal{L}$: loss function;
    $\sigma$: search radius;
    $\nu$: minimal loss score improvement;
    $\mathfrak{B}_{target}$: target hash threshold;
  procedure reset():
    target = $\emptyset$;
    best_hash = {seed: $\emptyset$, score: $+\infty$};
    best_loss = {seed: $\emptyset$, score: $+\infty$};
  upon <es.Start | $\theta_i$>:
    reset();
    target =  $\theta_i$
  upon <es.Stop>:
    reset();
  procedure es.follow($\theta_i$, $\mathfrak{S}_{\beta_k}$):
    If $\mathfrak{S}_{\beta_k}$ == $\emptyset$: 
      return $\theta_i$;
    Else:
      $\beta_k$ = RG($\mathfrak{S}_{\beta_k}$);
      return $\theta_i + \sigma \beta_k$;
  procedure es.evaluate($\theta_i$, $\mathfrak{S}_{\beta_k}$):
    $\mathcal{L}_{\theta_i + \sigma \beta_k}$, $v\mathcal{\boldsymbol{L}}_{\theta_i + \sigma \beta_k}$ = eval($f_{\theta_i + \sigma \beta_k}$);
    $\mathfrak{B}_{\theta_i + \sigma \beta_k}$ = hash($\mathcal{L}_{\theta_i + \sigma \beta_k}$, $v\mathcal{\boldsymbol{L}}_{\theta_i + \sigma \beta_k}$);
    return ($\mathfrak{B}_{\theta_i + \sigma \beta_k}$, $\mathcal{L}_{\theta_i + \sigma \beta_k}$);
  upon target != null:
    seed = rand();
    $\mathfrak{B}_{\theta_i + \sigma \beta_k}$, $\mathcal{L}_{\theta_i + \sigma \beta_k}$ = es.evaluate($\theta_i$, $\mathfrak{S}_{\beta_k}$);
    If $\mathfrak{B}_{\theta_i + \sigma \beta_k}<\mathfrak{B}_{target}$:
      best_hash = {seed: $\mathfrak{S}_{\beta_k}$, score: $\mathfrak{B}_{\theta_i + \sigma \beta_k}$};
      trigger <es.BestHash | $\theta_i$, $\mathfrak{S}_{\beta_k}$, $\mathfrak{B}_{\theta_i + \sigma \beta_k}$;
    If $\mathcal{L}(\theta_i + \sigma \beta_k) < \mathcal{L}(\theta_i) + \nu$:
      best_loss = {seed: $\mathfrak{S}_{\beta_k}$, score: $\mathcal{L}_{\theta_i + \sigma \beta_k}$};
      trigger <es.BestLoss | $\theta_i$, $\mathfrak{S}_{\beta_k}$, $\mathcal{L}_{\theta_i+\sigma\beta_k}$>;
\end{lstlisting}

\section{Permissioned Distributed Evolutionary Search}



The intuition behind the permissioned setting is to leverage the verifiability of proposed update vectors in ($1,\lambda$)-ES to re-use existing results in classical distributed algorithms. Specifically, given the iterative nature of ($1,\lambda$)-ES, we need the total order broadcast to be able to order the iterative steps between all correct workers.

\begin{lstlisting}[caption={Permissioned Distributed Search},
label=list:8-6,
captionpos=b,
float,
abovecaptionskip=\medskipamount,
mathescape=true,
label={PermissionedCode}]
Abstraction: 
  PermissionedEvolutionarySearch, instance ps
  
Uses:
  - EvolutionarySearcher, instance es,
      parameters ($\mathcal{L}$, $\sigma$, $\nu$, _)
  - TotalOrderBroadcast, instance tob
  
Interface:
  - Indication <ps.Output | point>:
    parameters found by the evolutionary search

Algorithm:
  Implements:
     PermissionedEvolutionarySearch, instance ps
  Parameters:
    $\mathcal{L}$: loss function;
    $\sigma$: search radius;
    $\nu$: loss threshold for update;
    $\theta_0$: starting point of the search;
    Z: number of search steps;
  upon <ps.Init>:
    target = $\theta_0$;
    steps = 0;
    trigger <es.Start | target>;
  upon <es.BestLoss | $\theta_i$, $\mathfrak{S}_{\beta_k}$, $\mathcal{L}_{\theta_i+\sigma\beta_k}$>:
    If $\theta_i$ == target And $\mathcal{L}(\theta_i + \sigma \beta_k) < \mathcal{L}(\theta_i) + \nu$:
      trigger <tob.Broadcast |
               ["ValidLoss", $\theta_i$, $\mathfrak{S}_{\beta_k}$]>;
  upon <tob.Deliver | 
        source_es ["ValidLoss", $\theta_i$, $\mathfrak{S}_{\beta_k}$]>:
    If $\theta_i$ == target:
      (_, $\mathcal{L}_{\theta_i + \sigma \beta_k}$) =
      es.evaluate($\theta_i$, $\mathfrak{S}_{\beta_k})$)
        // verify that the seed is valid indeed
      If $\mathcal{L}_{\theta_i + \sigma \beta_k} < \mathcal{L}_{\theta_i}+\nu$:
        target = es.follow($\theta_i$, $\mathfrak{S}_{\beta_k}$);
          // is actually $\theta_{i+1}=\theta_i + \sigma \beta_k$
        steps = steps + 1;
        If steps < Z:
          trigger <es.Start | target>;
        Else:
          trigger <es.Stop>;
          trigger <ps.Output | target>;
      Else:
        trigger <tob.Ban | source_es>;
          // optional penalty for misbehaving
\end{lstlisting}

\begin{theorem} \label{theorem:permissioned}
The algorithm in listing \ref{PermissionedCode} implements a machine learning consensus protocol that is Byzantine-resilient under the same assumptions as the Total Order Broadcast algorithm used, and is optimal with probability bound from above by $\frac{\Delta |\Pi|}{N T_{eval, average}}$, where $\Delta$ is the time needed to perform a Total Order Broadcast, $N$ - the expected number of tries to find a valid update seed and $T_{eval,average}$ is the average time needed by a worker to evaluate a candidate update seed.
\end{theorem}

\begin{proof}
The Total Order Broadcast ensures that the valid update seeds $\mathfrak{S}_{\beta_k}$ are delivered to all correct workers in the same order, after the workers were initialized to the same starting parameters value $\theta_0$. Assuming that a valid update seed exists $\forall i \in [0..Z-1]$, it will be eventually found and broadcasted by a correct worker.

Liveness: Each correct worker will eventually decide on the final $\theta_Z=\theta_0+\sum_{i=0}^{Z-1} \sigma \beta_{i, first}$, where $\sigma \beta_{i, first}$ is the update vector derived from the first valid update seed for point $\theta_i$.

Consistency: Thanks to the Total Order Broadcast, $\forall i \in [0..Z-1] \beta_{i, first}$ are same values for all workers and hence for each correct worker the final parameters of the model $\theta_Z=\theta_0+\sum_{i=0}^{Z-1} \sigma \beta_{i, first}$ are identical.

Extended Validity: By construction, the first correct worker to have its proposed seed successfully broadcast will have its update vector $\sigma\beta_{i, first}$ accepted. A seed that has not been successfully broadcasted cannot be accepted.

Probabilistic Optimality: By construction, at every step, upon the reception of a valid update seed $\mathfrak{S}_{\beta_{i,first}}$ through Total Order Broadcast, a correct worker will switch to searching for a valid update vector for the new parameters $\theta_i+1 = \theta_i + \sigma \beta_{i, first}$. The only way a better update at a given step becomes available without being broadcast first is if one becomes available during the total broadcast. The probability of that happening is proportional to the number of seed evaluations occurring before the broadcast completes times the probability of finding a seed above the threshold and better than the seed in the broadcast. The former is bound by the number of evaluations a worker pool can perform during the broadcast ($\frac{{|\Pi|} \Delta}{T_{eval,average}}$), whereas the second is bound by the chance of finding a valid seed, which, in case if the seed in broadcast is equal exactly to the validity threshold is $\frac{1}{N}$.

\end{proof}

\section{Permissionless Distributed Evolutionary Search}

\subsection{Proof-of-Work Mechanism for Probabilistic Consensus}

The probabilistic consensus algorithm through proof-of-work (PoW) was initially proposed in the Bitcoin blockchain whitepaper \cite{NakamotoBitcoin2008}, as a mechanism to achieve a probabilistic consensus through a leader election process tied to the amount of computational power actively committed to the election process. The principle of the election mechanism leverages the cryptographically secure hash function partial inversion. Based on the information provided by the prior leader election (often the hash of the prior block head), information to be broadcast by the next leader (often the root of the Merkle tree of transactions to be cleared), correct workers try to guess a random string (nonce) that once added to those two values would lead to a hash in the desired domain (for simplicity, $0 < \mathfrak{B} < \mathfrak{B}_{target} < \mathfrak{B}_{max}$). In turn, once a node finds a valid nonce, its leadership can be validated by other nodes by performing a single hash with the found nonce. This process is referred to as "mining" and each new leader election as a "block minting", and assuming sufficient time between leader elections to allow the previous block value to propagate ($\mathfrak{B}_{target}$ is adjusted based on the number of workers for that reason), ensuring an eventual election of a correct worker as a leader with high probability, assuming that the majority of computational power is controlled by correct workers \cite{2019NakomotoConsensusAnalysis, 2019NakamotoConsensus2}.

Unfortunately, the increasing popularity of PoW-based blockchains led to a combination of a large number of computationally powerful workers joining it and consequently to the difficulty threshold being increased to the point where PoW became a serious environmental problem \cite{BitcoinCarbonCellPress2019}. This led to heavy criticism of PoW consensus and other protocols - such as proof-of-stake \cite{ProofOfStakeOG2012} - to be promoted as less harmful alternatives for permissionless distributed consensus.

An alternative approach consisted in trying to highjack the proof of work to instead perform some useful work that would absorb computational resources independently of PoW-based blockchains. Such algorithms -  \textit{useful proof-of-work} (UPoW) - has unfortunately been hard to find, given the volume of computational power currently invested into PoW they would need to absorb and strict constraints on PoW to be usable: provably hard-to-find easy-to-verify updates, low communication complexity, and message weight and easily adjustable puzzle difficulty.

\subsection{Permissionless Distributed Evolutionary Search as Proof-of-Work}

However, given the ever-growing demand for computational power in machine learning, parameter space search problems are sufficiently common to leverage the computational power available to PoW consensus algorithms. Conversely, the distributed ($1,\lambda$)-ES seems to fit the constraints on the UPoWs, given that while hard to find, valid update vectors are straightforwards to validate and that communication overhead in-between iterative steps of ($1,\lambda$)-ES only contains  ($\mathfrak{S}_{\beta_k}$, $\mathcal{\boldsymbol{L}}_{\theta_k}$) - candidate update vector seed and associate loss.

To simplify the proofs and enable a direct mechanism for complexity adjustment, rather than using a valid update itself as a proof for leader election, we instead use the $\theta$-block score $\mathfrak{B}_{\theta_k}$, while propagating the best found valid update seed and associated loss ($\mathfrak{S}_{\beta_{k_{update}}}$, $\mathcal{L}_{\beta_{k_{update}}}$) with the same mechanisms as Merkel tree roots. Intuitively, this is a distributed equivalent of ($1,\lambda$)-ES with a set sampling population, except with the size decided by the expected candidate update samples between leader elections.

Given the variety of available blockchain protocols, we will abstract them away in the same we abstracted the total order broadcast in the permissioned setting and assume they implement an interface described in listing \ref{BlockchainInterface}.

\begin{lstlisting}[caption={Permissionless Distributed Search},
label=list:8-6,
captionpos=b,
float,
abovecaptionskip=\medskipamount,
mathescape=true,
label={PermissionlessCode}]
Abstraction: 
  PermissionlessEvolutionarySearch, instance ps
  
Uses:
  - EvolutionarySearcher, instance es, 
      parameters ($\mathcal{L}$, $\sigma$, $\nu$, $\mathfrak{B}_{target}$)
  - Blockchain, instance bl

Interface:
  - Indication <ps.Output | point>:
    parameters found by the evolutionary search
  - Procedure ps.processNewBlock([$\theta_i$, target, steps, 
    es.best_hash, es.best_loss])

Algorithm:
  Implements:
     PermissionlessEvolutionarySearch, instance ps
  Parameters:
    // same as in permissioned
  upon <ps.Init>:
    // same as in permissioned
  procedure processNewBlock([$\theta_i$, target, steps, 
    es.best_hash, es.best_loss]):
    If steps < Z:
      trigger <es.Start | target>;
    Else:
      trigger <es.Stop>;
      trigger <ps.Output | target>;
      trigger <bl.loadNext>;
  upon <es.BestLoss | $\theta_i$, $\mathfrak{S}_{\beta_k}$, $\mathcal{L}_{\theta_i+\sigma\beta_k}$>:
    If{$\theta_i$ == target And $\mathcal{L}(\theta_i + \sigma \beta_k) < \mathcal{L}(\theta_i) + \nu$:
      trigger <bl.sendValue |
               ["ValidLoss", $\theta_i$, $\mathfrak{S}_{\beta_k}$, $\mathcal{L}(\theta_i + \sigma \beta_k)$]>;
  upon <bl.deliverValue | 
        source_es ["ValidLoss", $\theta_i$,
        $\mathfrak{S}_{\beta_k}$, $\mathcal{L}_{declared}(\theta_i + \sigma \beta_k)$]>:
    If $\theta_i$ == target
        And $\mathcal{L}_{declared}(\theta_i + \sigma \beta_k)$ < es.best_loss[score]:
      (_, $\mathcal{L}_{validated}{\theta_i + \sigma \beta_k}$) =
      es.evaluate($\theta_i$, $\mathfrak{S}_{\beta_k})$);
        // verify that the sender is not lying
      If $\mathcal{L}_{declared}(\theta_i + \sigma \beta_k)$ == $\mathcal{L}_{validated}(\theta_i + \sigma \beta_k)$:
        es.best_loss = {seed: $\mathfrak{S}_{\beta_k}$, score: $\mathcal{L}_{\theta_i + \sigma \beta_k}$};
        trigger <bl.sendValue | 
          ["ValidLoss", $\theta_i$, $\mathfrak{S}_{\beta_k}$, $\mathcal{L}(\theta_i + \sigma \beta_k)$]>;
  upon <es.BestHash | $\theta_i$, $\mathfrak{S}_{\beta_k}$, $\mathfrak{B}_{\theta_i + \sigma \beta_k}$>:
    If $\theta_i$ == target And $\mathfrak{B}_{\theta_i + \sigma \beta_k} < \mathfrak{B}_{target}$:
      target = es.follow($\theta_i$, es.best_loss[seed]);
      steps += 1;
      trigger block = <bl.mintBlock | 
                       [$\theta_i$, target, steps,
                       es.best_hash, es.best_loss]>;
      trigger <bl.propagateBlock | block>;
      ps.ProcessNewBlock(block);
    upon <bl.deliverBlock |
          source_es [$\theta_i$, target, steps,
          source_es.best_hash, source_es.best_loss]>:
      If $\theta_i$==target:
        ($\mathfrak{B}_{\theta_i + \sigma \beta_k}$, _) =
          es.evaluate($\theta_i$, source_es.best_hash[seed]);
        If{$\theta_i$ == target And $\mathfrak{B}_{\theta_i + \sigma \beta_k} < \mathfrak{B}_{target}$:
          target = target;
          steps = steps;
          trigger <bl.propagateBlock | block>;
          ps.ProcessNewBlock(block);
\end{lstlisting}

\begin{lstlisting}[caption={Expected Blockchain interface},
label=list:8-6,
captionpos=b,
float,
abovecaptionskip=\medskipamount,
mathescape=true,
label={BlockchainInterface}]
Abstraction: 
  Blockchain, instance bl
  
Interface:
  - Procedure bl.sendValue: 
    worker proposes to its neighbours a value to be
    included in the next block
  - Procedure bl.deliverValue: 
    delivers a value proposed for inclusion into the
    next block from a neighbour
  - Procedure bl.mintBlock: 
    allows a worker to mint a new block that would
    include the best valid updates received
  - Procedure bl.propagageBlock: 
    allows a worker to suggests a newly found block
    to be propagated a neighbours
  - Procedure bl.deliverBlock: 
    delivers a block proposed for propagation from
    a neighbour
  - Procedure bl.loadNext: 
    loads the next distributed search task queued
    in blockchain
\end{lstlisting}

\begin{theorem}
Algorithm in Listing \ref{PermissionlessCode} is a valid proof-of-work and is a machine learning consensus protocol optimal with probability bound from above by $\frac{\Delta |\Pi|}{N T_{eval,average}}+e^{\Omega(\delta^2 g^2 d)}$ $\forall \delta > 0$ if $g^2\alpha > (1+\delta)\gamma$ after $d$ block added on top of the block minted during the step Z. $\Delta$ and is the time needed to propagate a block or a value, respectively\footnote{Given that the propagation of a value and block involves evaluating a candidate update, depending on the neighbor propagation topology, $Delta_{block/val}$ can be $O(|\Pi| \cdot T_{eval,average})$, $O(log(|\Pi|) \cdot T_{eval,average})$ or $O(T_{eval,slowest})$. For the sake of generalizability, we keep the same notation as previously}, $N$ - the expected number of tries to find a valid update seed, and $T_{eval,average}$ is the average time needed by a worker to evaluate a candidate update seed, $\alpha$ and $\gamma$ - collective minting rates of correct and faulty nodes and $g=e^{-\alpha\Delta}$, the propagation delay penalty for correct nodes. 
\end{theorem}

\begin{proof}
The algorithm in the listing \ref{PermissionlessCode} is a valid proof-of-work because the block minting mechanism is equivalent to a partial inversion of a cryptographic function with an unavoidable loss function evaluation overhead.

Since the algorithm in the listing \ref{PermissionlessCode} is a valid proof-of-work, the Nakamoto consensus regarding the block propagated at the step Z of ps after $d$ blocks were added on top of it will not change with probability $1-e^{\Omega(\delta^2 g^2 d)}$ for any $\delta > 0$ as long as $g^2\alpha > (1+\delta)\gamma$, as per \cite{2019NakomotoConsensusAnalysis, 2019NakamotoConsensus2}.

The Nakamoto consensus blockchain blocks are available to all correct workers and are ordered in a unique way for all correct workers. By replacing Total Order Broadcast in the proof of by the blockchain segment read containing blocks corresponding to steps 0 to Z, the proof for \ref{theorem:permissioned} applies.





\end{proof}

The intuitive explanation of proof is that the algorithm in listing \ref{PermissionlessCode} will fail to register the best random seed with the best candidate update loss in only two cases. First, if the blockchain forked and the bock with the best candidate update random seed finding itself on the dead branch. This case occurs with the probability $1-e^{\Omega(\delta^2 g^2 d)}$. Second, if the candidate update seed with the best loss is found within the time $\delta$ from the block update, accounted for by term $\frac{\Delta |\Pi|}{N T_{eval,average}}$. While this is possible if the task supplied is too easy for the size of the blockchain, there is likely a tighter bound, given that if many valid update seeds were found as a single block was mined, the last one is not necessarily the one that will reach the best loss.








\section{Discussion and Related Work} \label{relatedWork}

\subsection{Byzantine-Resilient Distributed Machine Learning} \label{BRDML}

As we mentioned in the introduction, while multiple mechanisms to distribute machine learning algorithms were proposed, the only one accounting for Byzantine faults is the ($\alpha,f$)-Byzantine Resilient learning proposed by \cite{Mehdi01, Mehdi01Bis}. While initially formulated in the context of federated learning with a centralized parameter aggregation server, it has been developed to allow byzantine fault tolerance in a general distributed setting, under realistic assumptions \cite{Mehdi2021Jungle}.

Compared to our approach, a downside of the ($\alpha,f$)-Byzantine Resilience is their dependence on direct gradients and a requirement for the model to be sufficiently small for the gradient vector sharing to not become a bottleneck. Modern ML and AI models are often sufficiently large for data transfer time to not be negligible and need to be synchronized often enough to avoid parameter drift between models. For instance, the GPT-3 generative language model has several hundreds of GBs of parameters, and even when split into single attention heads requires each parameter synchronization to transfer several GBs of parameter values.

Conversely, an advantage of ($\alpha,f$)-Byzantine Resilience is its ability to train on the data that's different on each worker node, assuming the data is uniform across workers. While this assumption is not necessarily true in all settings, our distributed ($1,\lambda$)-ES does not natively support the data distribution\footnote{Given the linear nature of the loss wrt to data, it is possible to design variants of distributed ($1,\lambda$)-ES that would require a sufficient number of nodes to confirm that a candidate update vector achieves an acceptable loss improvement on their data as well. This would however require additional assumptions and a more restrictive learning setting and is out of the scope of this paper.
}. The main problem addressed by distributing ($1,\lambda$)-ES is the difficulty to find a single valid update at each step of the optimization task, particularly in cases where the parameter space is in a very high dimension, leading to long valid update vector search, even when single candidate update vector evaluation is itself fast. This setting is specifically the one in which EAs have been successfully applied in modern machine learning \cite{CluneMiikkulainen2019, 2017UberGeneticAlgos, JeffCluneNeurIPS2018}. Notably, the permissionless UPoW ($1,\lambda$)-ES version only makes sense in that setting, complemented with a high iteration number to ensure that most of the communication is carrying only candidate update seeds and associated loss values and minimize training data transfer.

Finally, while there are other distributed approaches to gradient-free optimization, notably for Support Vector Machines (SVMs) \cite{Vapnik1995SVMsEnglish, ConsensusSVM2010} and the Genetic Algorithm \cite{Goldberg1989GAs, DistributedGA1995}, none to our knowledge allow for byzantine faults.



\subsection{Useful proof of work} \label{PoUWREview}

Given that the computational and energy costs of PoW consensus were already well-known by mid-2010s \cite{BitcoinCarbonCellPress2019}, multiple attempts were made to leverage the PoW for useful work. The first proposal was made in 2017, leveraging a set of problems in computational geometry for which the search of a solution is $~O(n^2)$ hard and verification is $~O(n)$ hard \cite{Ball2017ProofOfUsefulWork}. Unfortunately, insufficient demand and lack of difficulty adjustment mechanism meant that this PoW was impractical. The same year another set of problems - partitioned linear algebra on very large non-sparse matrices was proposed by \cite{Shoker2017PODCProofOfExercise}. Unfortunately, that PoW did not take either, given the rarity of problems involving such matrices and amounts of data (TBs) that would need to be transferred in the process.
Finally, still the same year, a highly general framework for turning any computationally intensive task into PoW challenges has been proposed \cite{ResourceEfficientMining2017}. Relying on the trusted hardware - Intel SGX - it initially showed a great performance but was rapidly rendered obsolete by the rarity of hardware supporting Intel SGX and then the demise of Intel SGX in the wake of Spectre vulnerabilities \cite{CVE-2017-5753, CVE-2017-5715}.

The next iteration of the search for a useful proof of work started in 2019, focusing on NP-hard problems, notably the traveling salesman problem and machine learning. While some success has been achieved in by using TSP in the context of the container ship sailing route optimization \cite{Haouari2022ContainerTSPPoUW}, despite being NP-hard TSP has several probabilistic heuristics available and no initial estimation of hardness for a specific problem, making it non-strategy-proof and hence not a suitable PoW for blockchain purposes.

Machine learning PoW has been attempted as well, however, all the approaches we are aware of tried to use gradient-based machine learning and ran into the fundamental issue of non-verifiability of gradient calculation, leading them to waste resources through replication or to leave their PoW non-strategy-proof and often to have the model training itself to be vulnerable to adversaries. Examples of such approaches are Proof of Learning (PoLe) \cite{Liu2021ProofOfLearning}, which is essentially a race to a predefined accuracy, and model hyperparameter sweep PoW \cite{Baldominos2019CoinAI}, neither containing replication, leaving ML model vulnerable to attackers, and not being strategy-proof, given that a worker with more computational resources or a good heuristic could consistently "win" each of those competitions.

Proof of Search \cite{Shibata2019ProofOfSearch} occupies an interesting spot among the proposed useful PoW in that it doesn't propose a useful PoW per se, but rather a way to transform suitable useful PoW into a blockchain. In that, it is complementary to our approach, since it provides a blueprint to implement the blockchain interface described in the listing \ref{BlockchainInterface}.

Perhaps most relevantly to the evolutionary search aspect of our algorithm, two evolutionary proofs of useful work have been proposed, both using the genetic algorithm \cite{Syafruddin2019EvolutionaryProofOfWork, Bizzaro2020ProofOfEvolution}. One focuses on solving the TSP problem by using a genetic algorithm directly, whereas the other one tries to create a framework similar to the Proof of Search, but using genetic algorithms specifically and applied to NP-hard problems such as TSP and Knapsack. Not only the approaches are vulnerable to the issues mentioned in the context of TSP blockchains mentioned above, but the genetic algorithm also implies a collaborative phase of parameter mixing during the "chromosome" "cross-over"\footnote{As we mention in the introduction, we use Genetic Algorithm only to designate EAs that has "chromosome" and "cross-over" phases, consistently with the nomenclature introduced in \cite{2015EvolutionStrategiesReview}}, which runs against the competitive nature of the Proof-of-Work and adds a layer of communication complexity. In addition to that, both approaches require messages between workers to carry whole parameter update vectors, adding a communication bottleneck for larger models.

\section{Conclusion}

In this paper, we present a new definition of distributed machine learning consensus - the \textit{Model-Consensus}, that generalizes the previously proposed ($\alpha$, $f$)-Byzantine Resilience and is applicable both in differentiable and non-differentiable settings.

We then present two distributed versions of the ($1,\lambda$)-Evolutionary Search algorithm, both reaching a model-consensus in a gradient-free setting. One leveraged the classical distributed algorithms abstraction of Total Order Broadcast to achieve a consensus in a permissioned setting, whereas the other used the Proof-of-Work leader election consensus to achieve the same result in a permissionless setting.

To our knowledge, our model-consensus definition is the first definition of Byzantine resilient consensus in machine learning that covers both gradient-free and gradient-based learning, while generalizing the previously proposed ($\alpha,f$)-Byzantine Resilience and allowing for direct compatibility with the classical distributed algorithms consensus definition.

To our knowledge, the two algorithms we propose are the first Byzantine-resilient gradient-free learning algorithms and the first byzantine-resilient evolutionary algorithms. Likewise, our permissionless distributed evolutionary search algorithm is the first useful proof-of-work algorithm that minimizes the overhead compared to traditional proof-of-work.

While proposed algorithms work for any black-box optimization problems, they are most suited for high-dimensional optimization problems where the evaluation of a single updated parameter set is quick, but the sheer number of dimensions makes the search for a valid update excessively slow for a single worker, with a notable application being the neuroevolution of large ANNs. We foresee that such a setting is of particular relevance to multipurpose super neural networks, such as PathNet \cite{Pathnet2017DeepMind}, as well as for conversational agents derived from LLMs \cite{SeeKeR2022Facebook, Sparrow2022Google}. 

Finally, more accessible and reliable byzantine-resilient machine learning allows a variety of entities to poll together their computational resources to train models they could not have trained individually, which has the potential of democratizing the state-of-the-art ML research in a non-differentiable setting.

\section{Acknowledgments}

We would like to thank the armasuisse - Cyber-Defence (CYD) Campus for the Distinguished Post Doctoral Fellowship supporting AK, as well as Fabien Salvi (EPFL) for the technical support regarding the computational infrastructure organization, and France Faille (EPFL) for the administrative support.

\bibliographystyle{ACM-Reference-Format}
\bibliography{sample-base}










\end{document}